\documentclass{egpubl}
\usepackage{wait2022}

\usepackage{graphicx} 
\usepackage{bmpsize}
\usepackage[utf8]{inputenc}
\usepackage{amsmath}
\usepackage{amssymb}

\usepackage{caption}
\usepackage{subcaption}

\newtheorem{theorem}{Theorem}
\newtheorem{lemma}{Lemma}

\ConferencePaper      

\title{A generalization of tangent-based implicit curves}

\author[Á. Sipos]
       {Ágoston Sipos
        \\
        Department of Control Engineering and Information Technology, Budapest University of Technology and Economics
       }

\PrintedOrElectronic

\usepackage{t1enc,dfadobe}

        \let\tt=\ttfamily

\begin{document}

\maketitle

\begin{abstract}

An approach to defining quadratic implicit curves is to prescribe two tangent lines and a secant line going through the points of tangency.\cite{Liming:1947} This paper will show that this method can be generalized to a higher number of tangents, resulting in higher degree curves.

\end{abstract}

\section{Introduction}

Implicit curves are widely used in approximation and design. Their advantages include: (1) they can very efficiently be intersected with parametric curves; (2) it is easy to classify points based on which side of the curve they are; (3) an offset curve can be simply defined by adding a constant to the equation. However, they also have disadvantages like: (1) generating points along the curve is hard; and (2) implicit curves may - and often do - have multiple unintended branches.

It is not necessarily intuitive to design with implicit curves. Generally, they are a mathematical formula in the form $f(x,y) = 0$, but the function $f$ usually does not have a natural geometric interpretation (except in the case of simple objects like lines and circles).

For quadratic curves, Liming\cite{Liming:1947} proposed a method for representing the curve with tangent and secant lines. This means not having to work with the quadratic function's geometrically meaningless coefficients, instead, the implicit curve can be evaluated at any point by evaluating those lines and substituting them into a fixed formula.

This paper is going to present a way to generalize this approach to having a higher number of prescribed tangent lines defining the curve, which subsequently will be of a higher degree.

\section{Previous work}

Liming's method for conic (quadratic) curves\cite{Liming:1947} works as follows (see also Figure \ref{fig:conic}):

\begin{enumerate}
    \item Let $L_1, L_2 : \mathbb{R}^2 \rightarrow \mathbb{R}$ be two lines in implicit form, that are two (distinct) tangents to the desired curve
    \item Let $C : \mathbb{R}^2 \rightarrow \mathbb{R}$ be a secant line to the curve, going through the two points of tangency
    \item Then, $\forall \lambda \in (0;1) : $ the curve \begin{equation} Q(x,y) = (1-\lambda) \cdot L_1(x,y) \cdot L_2(x,y) - \lambda \cdot C^2(x,y) \end{equation} is a nonsingular quadratic curve, fulfilling the tangential conditions.
\end{enumerate}

\begin{figure}
    \centering
    \includegraphics[width=\linewidth]{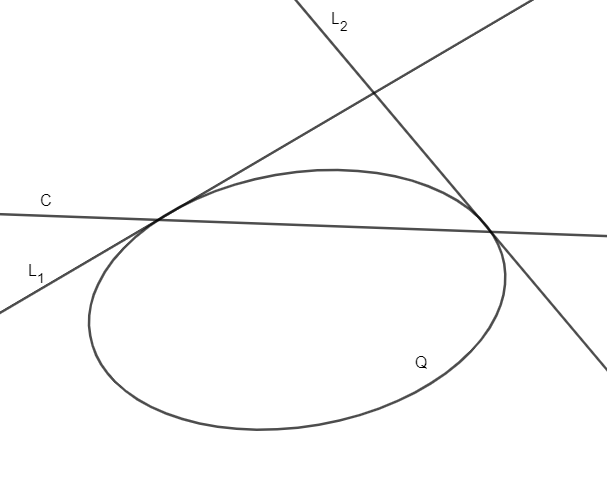}
    \caption{Example of Liming's method}
    \label{fig:conic}
\end{figure}

There are multiple similar schemes for 3D surfaces, like the functional spline\cite{Li:1990} which is a direct generalization of Liming's method for an arbitrary number of tangential surfaces. It is however not in all situations clear how to describe the \emph{transversal surface} which is the analogy of the secant line in Liming's.

Another 3D method is the \emph{\mbox{I-patch}}\cite{Varady:2001} (Figure \ref{fig:ipatch}) which is defined by $n$ \emph{ribbon} (tangential) surfaces ($R_i$) and $n$ \emph{bounding} surfaces ($B_i$). The patch equation is
\begin{equation}
    I = \displaystyle\sum_{i=1}^n \Big( w_i R_i \displaystyle\prod\limits_{\substack{j=1 \\ j\neq i}}^n B_j^2 \Big) + w_0 \displaystyle\prod_{j=1}^n B_j^2.
\end{equation}
\begin{figure}
    \centering
    \includegraphics[width=\linewidth]{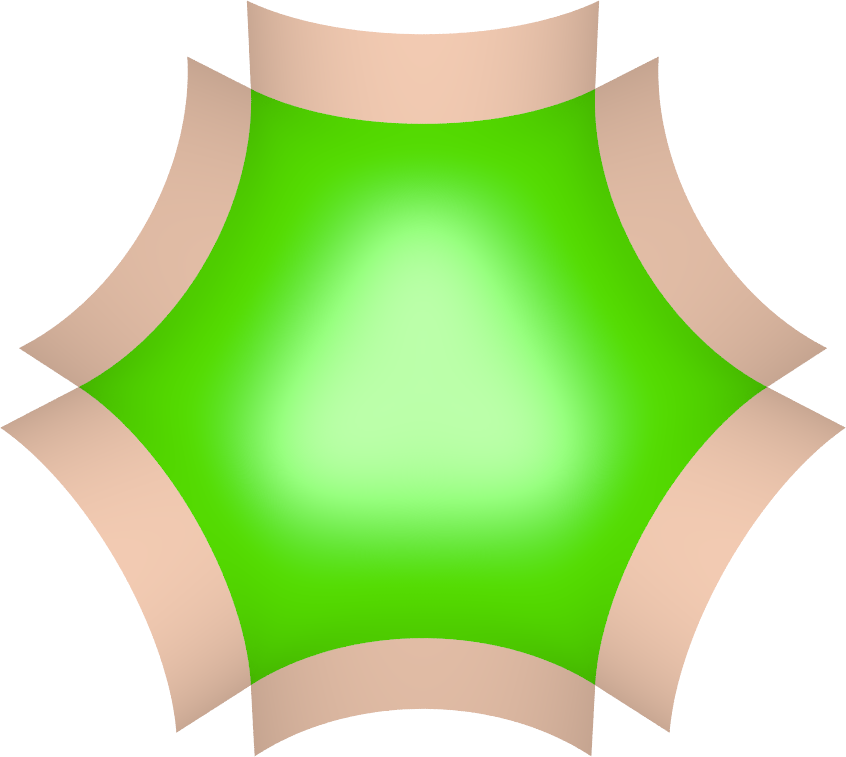}
    \caption{A six-sided I-patch in 3D}
    \label{fig:ipatch}
\end{figure}
The I-patch also has a different form, the \emph{faithful \mbox{I-patch}}\cite{Sipos:2022}:
\begin{equation}
    \hat{I} = \dfrac{I}{\displaystyle\sum_{i=1}^n w_i \prod\limits_{\substack{j=1 \\ j\neq i}}^n B_j^2}.
\end{equation}
It is advantageous for multiple reasons, as it does not have isolated points, and has a more even distance metric. However, the first goal is also achieved via using a  \emph{normalized I-patch}:
\begin{equation}
    I^* = \dfrac{I}{\displaystyle\sum_{i=1}^n \prod\limits_{\substack{j=1 \\ j\neq i}}^n B_j^2}.
\end{equation}

\section{Notes on Liming's method}

It is easy to see that the Liming curve is a one-sided 2D \mbox{I-patch}, with $R_1 = L_1 \cdot L_2$, $B_1 = C$ and $\frac{w_1}{w_0} = \frac{1-\lambda}{\lambda}$.

We will also utilize the following important property of Liming's method. Although intuitive, no proof was found in the literature, so we will include one.

\begin{lemma}
\label{lem:quad}
Let $Q$ be a quadratic implicit curve, $L_1, L_2$ two distinct tangent lines of it, $C$ the secant through the points of tangency. Then, $\exists \lambda \in (0;1), \omega \in \mathbb{R}$ s.t. 
\begin{equation}
(1-\lambda) \cdot L_1 \cdot L_2 - \lambda \cdot C^2 \equiv \omega \cdot Q.
\end{equation}
\end{lemma}

\begin{proof}
Let $(x,y)$ be a point s.t. $Q(x,y) = 0$, $L_1(x,y) \neq 0$, $L_2(x,y) \neq 0$. Let 
\begin{equation}
    \lambda := \dfrac{L_1(x,y)\cdot L_2(x,y)}{L_1(x,y)\cdot L_2(x,y) + C^2(x,y)}
\end{equation}
meaning that this point is on our curve as well.

A general quadratic implicit curve is of the form
\begin{equation}
    a x^2 + b x y + c y^2 + d x + e y + f,
\end{equation}
having 6 coefficients, although multiplying them with the same scalar does not change the curve.

Fitting a curve on the two tangential conditions, and the additional requirement of interpolating $(x,y)$ means a homogeneous linear system of 5 equations (see Appendix \ref{app:eq}). This means that the dimension of the solution space is one, so any two solutions are each other multiplied by a scalar.

\end{proof}

\section{Curves defined by four tangents}

Let $L_i \, (i = 1..4)$ be four distinct tangent lines, $\mathbf{p}_i \, (i=1..4)$ the points of tangency.

\begin{theorem}
Let $C_1$ be a line through $\mathbf{p}_1$ and $\mathbf{p}_2$, $C_2$ through $\mathbf{p}_3$ and $\mathbf{p}_4$. Then, the family of two-sided I-patches \begin{equation} I_{\mathbf{w}} := w_1 \cdot L_1 \cdot L_2 \cdot C_2^2 + w_2 \cdot L_3 \cdot L_4 \cdot C_1^2 + w_0 \cdot C_1^2 \cdot C_2^2 \end{equation} fulfills the tangential conditions. Moreover, if there exists a quadratic curve satisfying the constrained tangents, the normalized I-patch $I^*$ reproduces it with well-chosen coefficients.

\end{theorem}

\begin{proof}
The first statement trivially follows from the I-patch properties: for example in $\mathbf{p}_1$, using that $L_1(\mathbf{p}_1) = 0$ and $C_1(\mathbf{p}_1) = 0$; 
\begin{equation}
    I_{\mathbf{w}}(\mathbf{p}_1) = 0+0+0 = 0
\end{equation} and 
\begin{multline}
    \nabla I_{\mathbf{w}}(\mathbf{p}_1) = c_1(\mathbf{p}_1) \cdot \nabla L_1(\mathbf{p}_1) + \mathbf{v}_1(\mathbf{p}_1) \cdot L_1(\mathbf{p}_1) +\\
    + c_2(\mathbf{p}_1) \cdot 2 \cdot C_1(\mathbf{p}_1) \cdot \nabla C_1(\mathbf{p}_1) + \mathbf{v}_2(\mathbf{p}_1) \cdot C_1^2(\mathbf{p}_1),
\end{multline}
where $c_1, c_2$ are real-valued, $\mathbf{v}_1, \mathbf{v}_2$ are vector-valued functions. Using the constraints, we get
\begin{equation}
    \nabla I_{\mathbf{w}}(\mathbf{p}_1) = c_1(\mathbf{p}_1) \cdot \nabla L_1(\mathbf{p}_1) + 0 + 0 + 0,
\end{equation}
meaning that the point is on the curve, and the gradient's direction is the same as the line's gradient's, so $L_1$ is indeed a tangential line. $\square$

For the second statement, let $Q$ be the target quadratic curve,
\begin{align}
    R_1 &:= L_1 L_2\\
    R_2 &:= L_3 L_4
\end{align}
\begin{equation}
    S_{i,\omega,\lambda} := \omega ((1-\lambda) R_i + \lambda C_i^2)
\end{equation}

Due to Lemma \ref{lem:quad}, $\exists \omega_1, \lambda_1 : S_{1,\omega_1,\lambda_1} \equiv Q$, similarly $\exists \omega_2, \lambda_2 : S_{2,\omega_2,\lambda_2} \equiv Q$

Now, on one hand
\begin{multline}
    S_{1,\omega_1,\lambda_1} C_2^2 + S_{2,\omega_2,\lambda_2} C_1^2 =\\
    = \omega_1 (1-\lambda_1) R_1 C_2^2 + \omega_1 \lambda_1 C_1^2 C_2^2 +\\
    + \omega_2 (1-\lambda_2) R_2 C_1^2 + \omega_2 \lambda_2 C_1^2 C_2^2
\end{multline}

On the other hand
\begin{equation}
    S_{1,\omega_1,\lambda_1} C_2^2 + S_{2,\omega_2,\lambda_2} C_1^2 = Q \cdot (C_1^2 + C_2^2)
\end{equation}

This means that with
\begin{align}
    w_1 &= \omega_1 (1-\lambda_1)\\
    w_2 &= \omega_2 (1-\lambda_2)\\
    w_0 &= \omega_1 \lambda_1 + \omega_2 \lambda_2;
\end{align}

\begin{equation}
    I_{[w_1,w_2,w_0]} \equiv Q \cdot (C_1^2 + C_2^2), 
\end{equation}

so $I^*_{\mathbf{w}} \equiv Q$
\end{proof}

\section{Examples}

Four-tangent curves of course depend on the $w_i$ weights, so for a tangent configuration, many different curves can be obtained.

In the figures, blue lines represent the tangent lines, red lines are the cutting lines and the resulting curve itself is shown in purple.

\begin{figure}
    \centering
    \includegraphics[width=\linewidth]{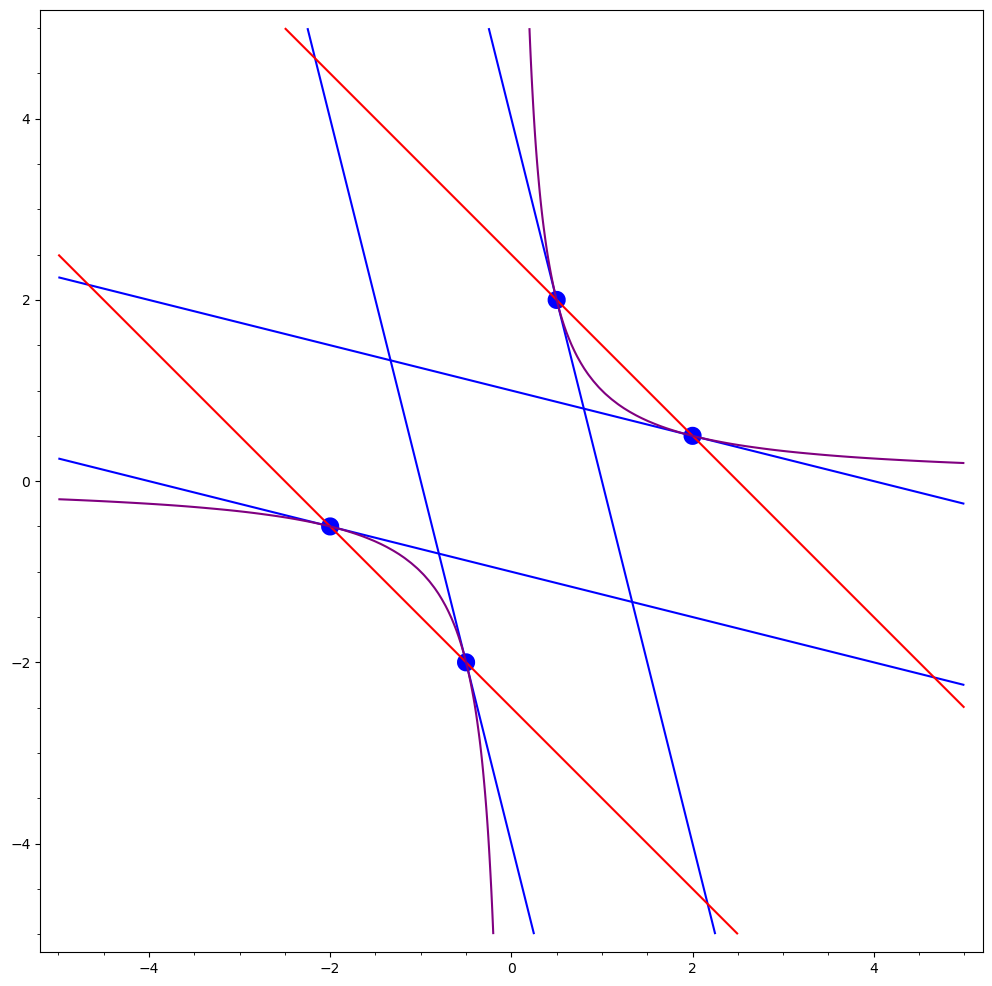}
    \caption{Reproducing a hyperbole}
    \label{fig:hyperbole}
\end{figure}

\begin{figure}
    \centering
    \includegraphics[width=\linewidth]{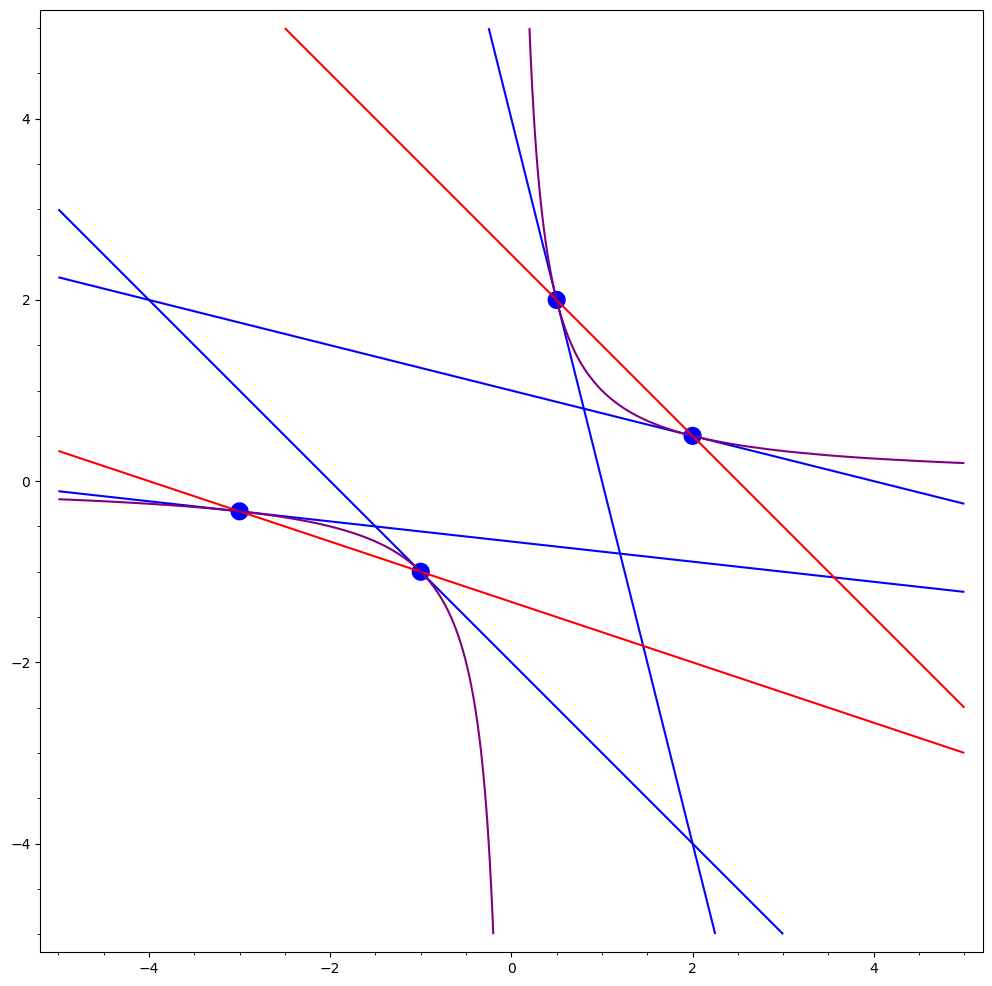}
    \caption{Reproducing a hyperbole with different tangents}
    \label{fig:hyperbole2}
\end{figure}

In Figures \ref{fig:hyperbole} and \ref{fig:hyperbole2} you can see the reproducibility of a hyperbole by choosing four points and four tangents of it, and setting in the first case $w_1=w_2=4/9$, $w_0=32/81$. In the second example, the correct weights are $w_1 = 3/4, w_2 = 4/9$ and $w_0 = 985/1296$.

\begin{figure}
     \centering
     \begin{subfigure}[b]{0.48\linewidth}
         \centering
         \includegraphics[width=\linewidth]{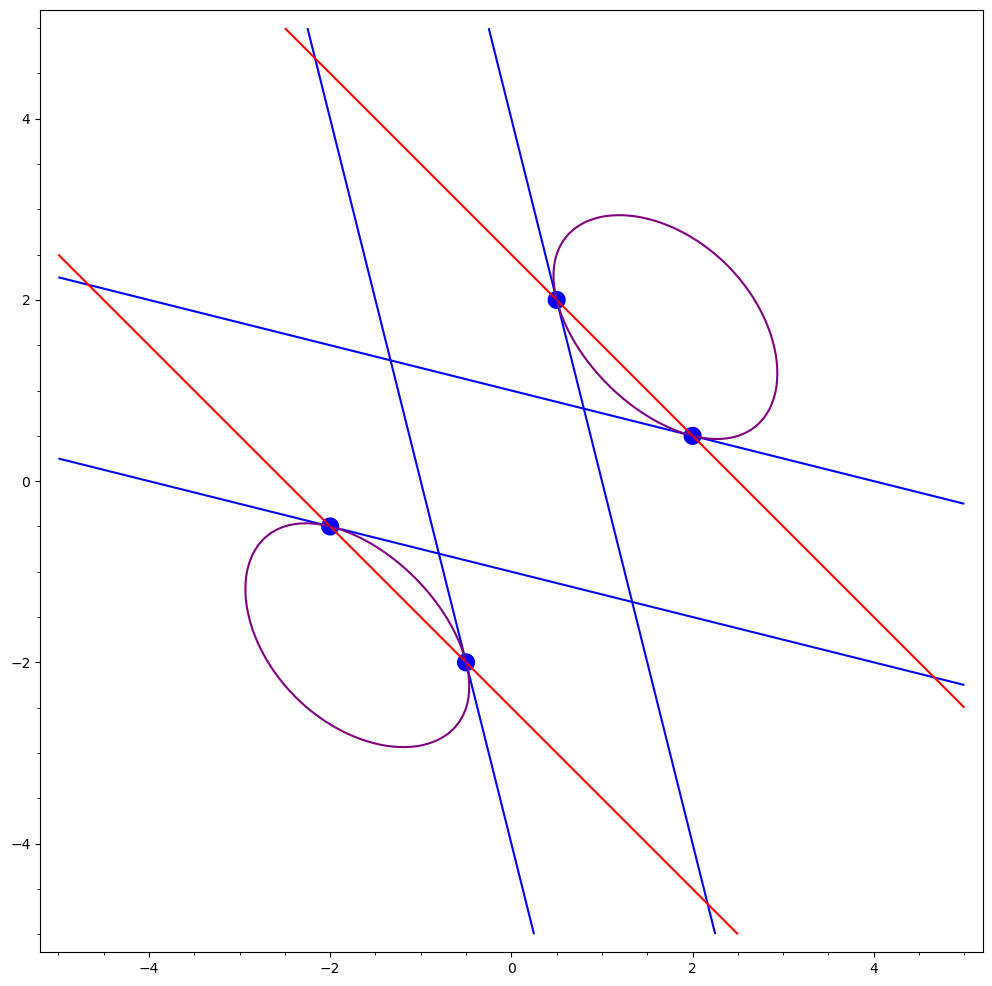}
         \caption{$w_1=0.5, w_2=0.5$}
     \end{subfigure}
     \hfill
     \begin{subfigure}[b]{0.48\linewidth}
         \centering
         \includegraphics[width=\linewidth]{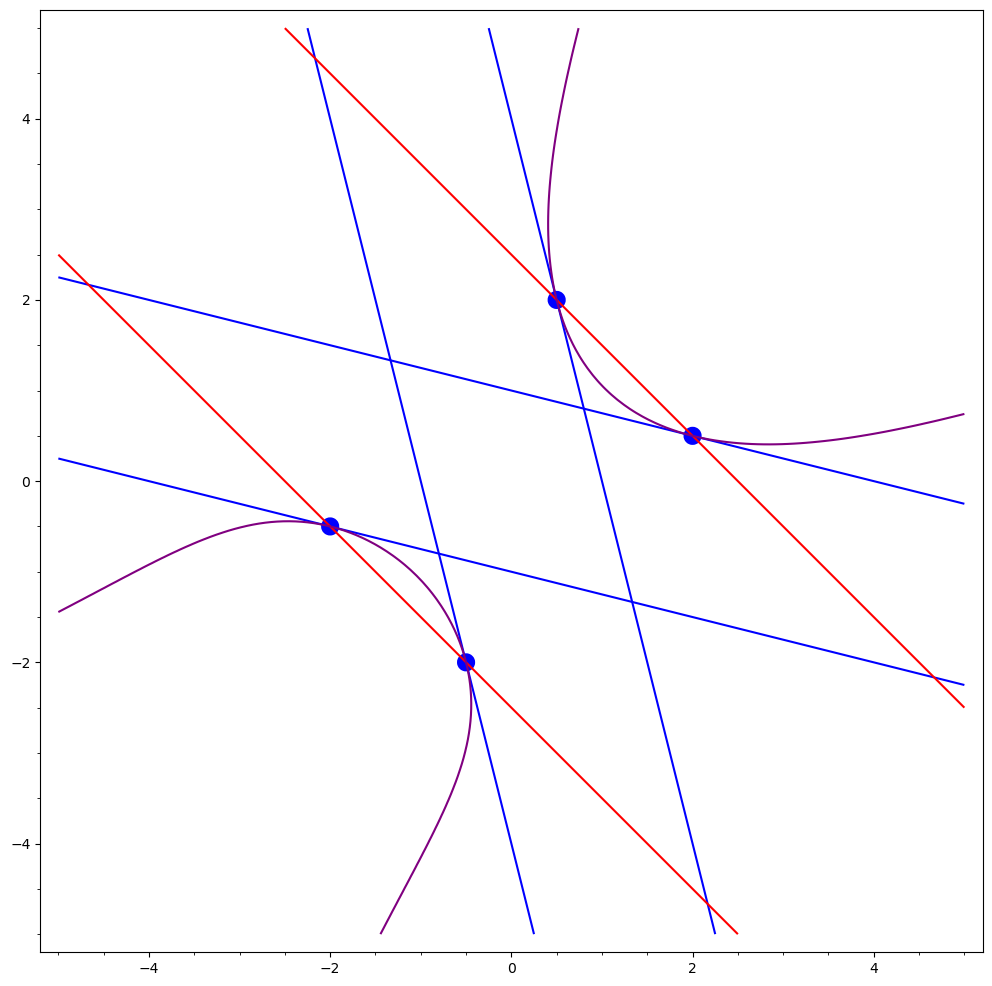}
         \caption{$w_1=0.5,w_2=1.0$}
     \end{subfigure}
     \hfill
     \begin{subfigure}[b]{0.48\linewidth}
         \centering
         \includegraphics[width=\linewidth]{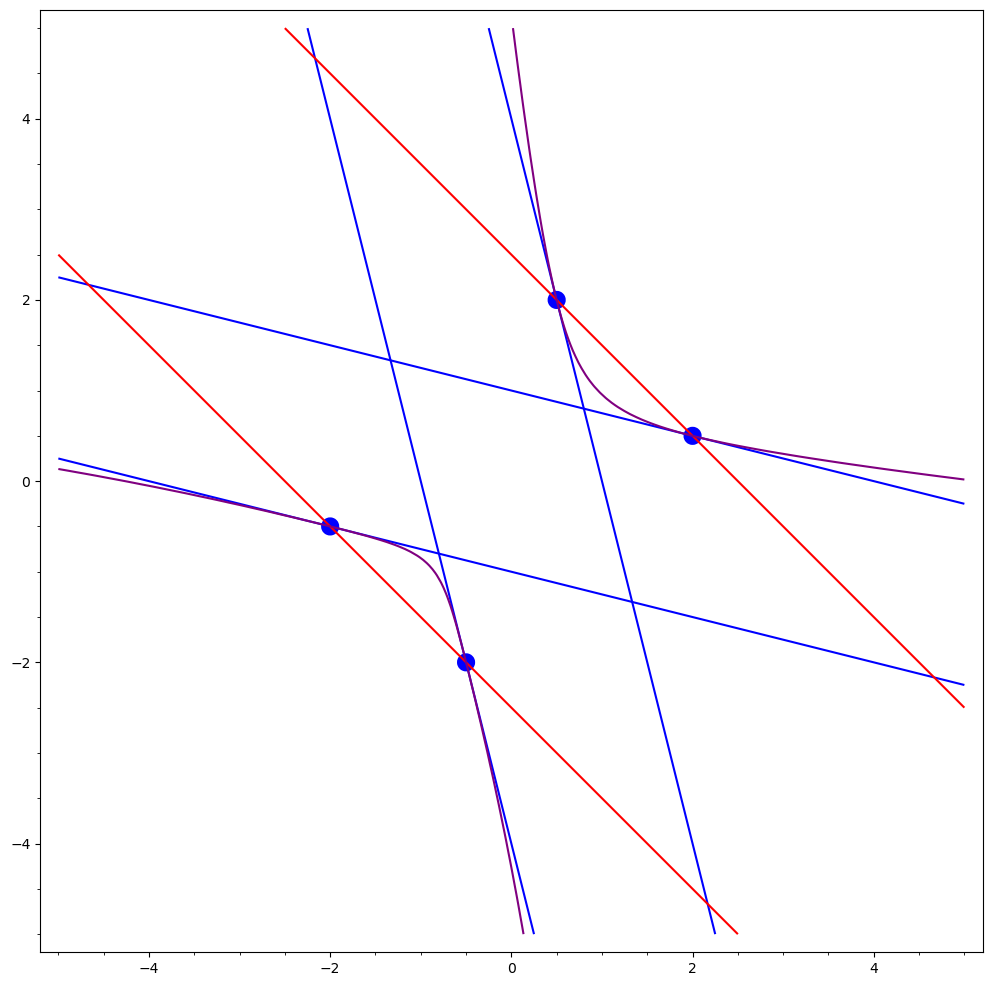}
         \caption{$w_1=0.8,w_2=2.0$}
     \end{subfigure}
     \hfill
     \begin{subfigure}[b]{0.48\linewidth}
         \centering
         \includegraphics[width=\linewidth]{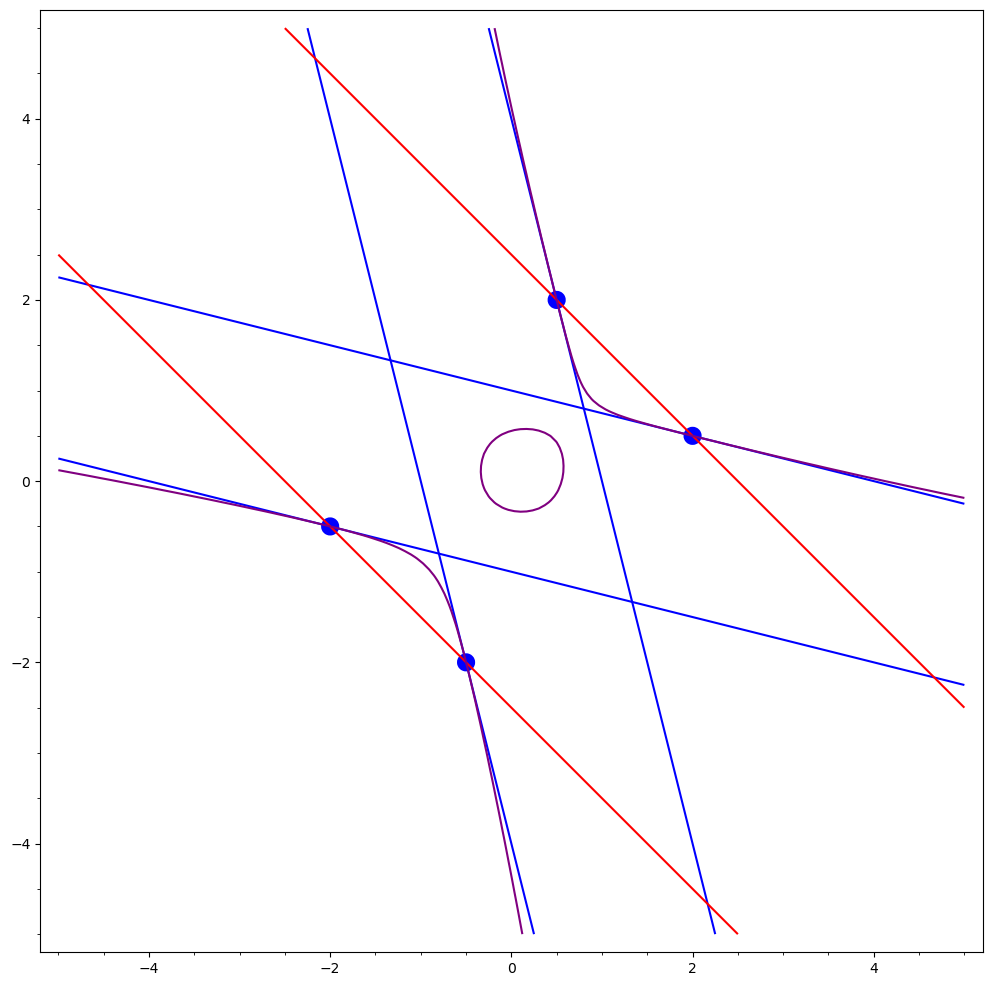}
         \caption{$w_1=2.0,w_2=1.5$}
     \end{subfigure}
        \caption{Different coefficients on the same data as Figure \ref{fig:hyperbole}}
        \label{fig:alternativ_hyperbole}
\end{figure}

In Figure \ref{fig:alternativ_hyperbole} it can be seen that changing the weights can greatly change the shape of the curve.

\begin{figure}
     \centering
     \begin{subfigure}[b]{0.48\linewidth}
         \centering
         \includegraphics[width=\linewidth]{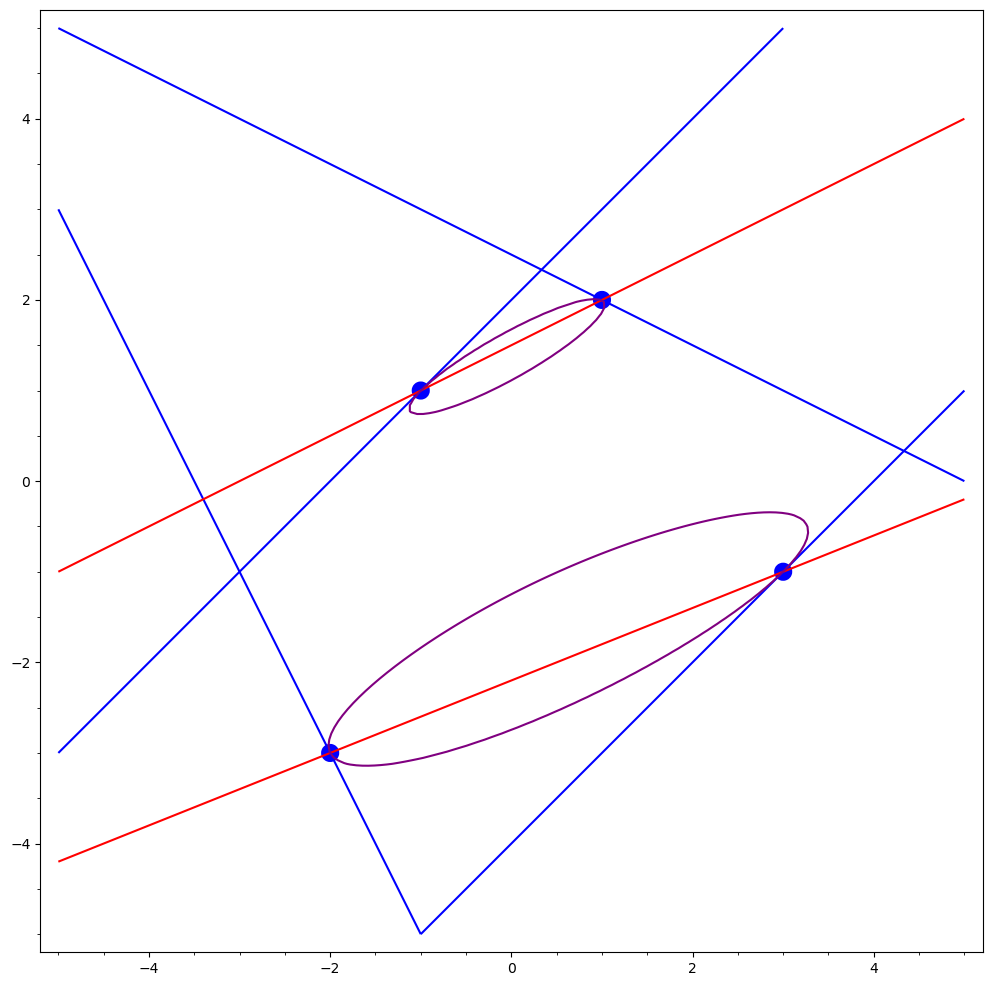}
         \caption{$w_1=0.2, w_2=1.2$}
     \end{subfigure}
     \hfill
     \begin{subfigure}[b]{0.48\linewidth}
         \centering
         \includegraphics[width=\linewidth]{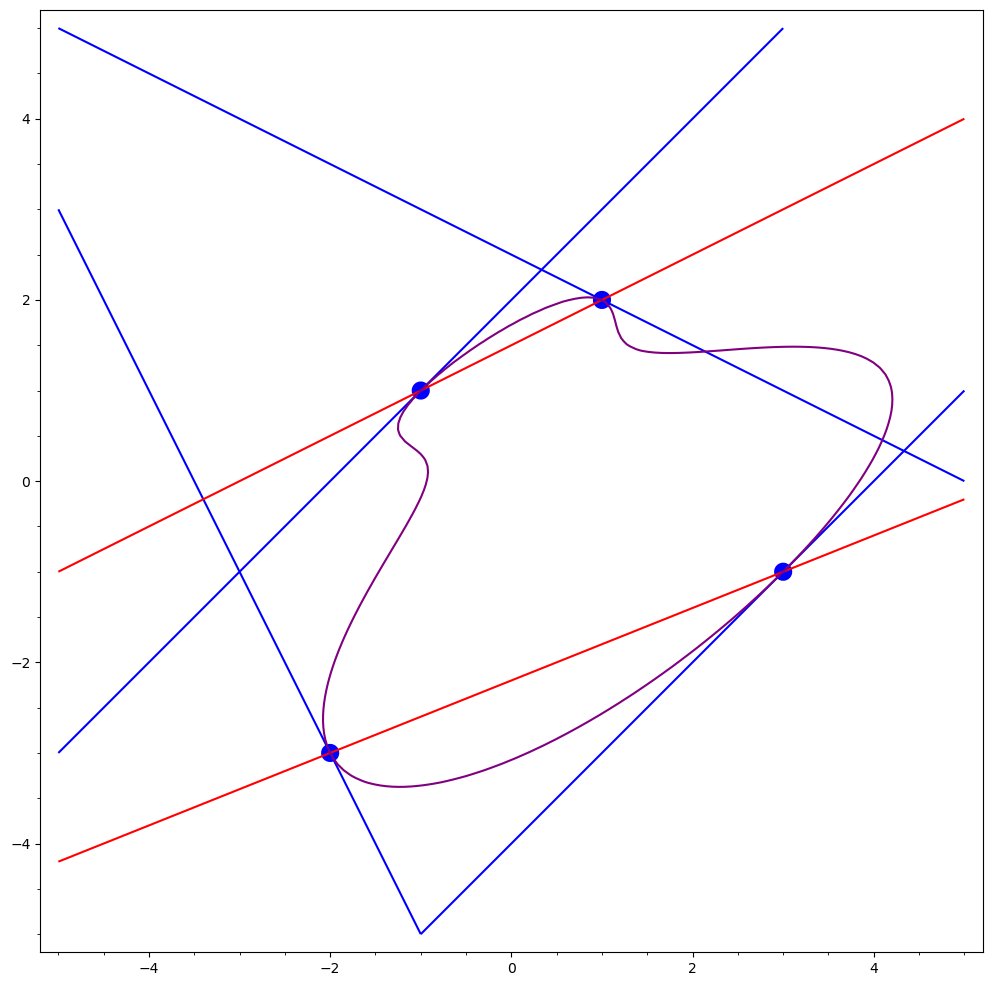}
         \caption{$w_1=0.2,w_2=4.6$}
     \end{subfigure}
     \hfill
     \begin{subfigure}[b]{0.48\linewidth}
         \centering
         \includegraphics[width=\linewidth]{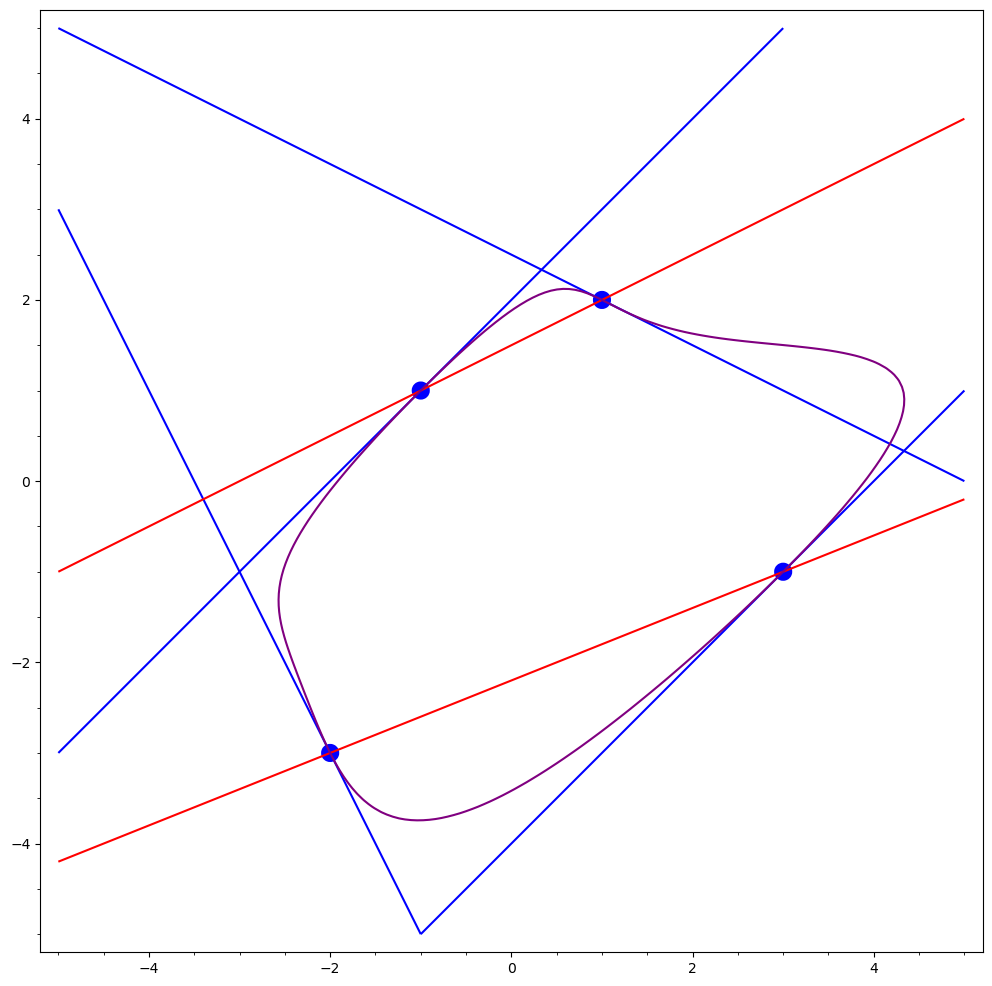}
         \caption{$w_1=1.0,w_2=6.0$}
     \end{subfigure}
     \hfill
     \begin{subfigure}[b]{0.48\linewidth}
         \centering
         \includegraphics[width=\linewidth]{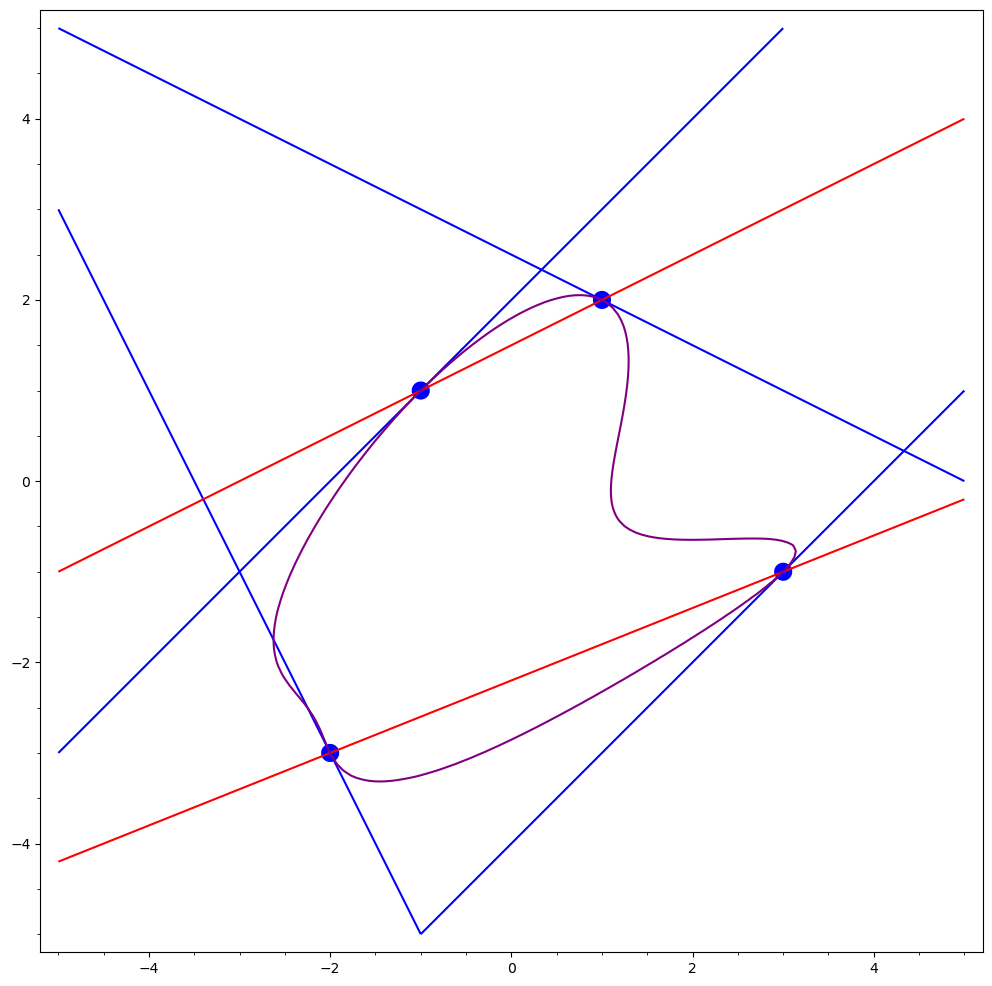}
         \caption{$w_1=1.2,w_2=0.4$}
     \end{subfigure}
     \hfill
     \begin{subfigure}[b]{0.48\linewidth}
         \centering
         \includegraphics[width=\linewidth]{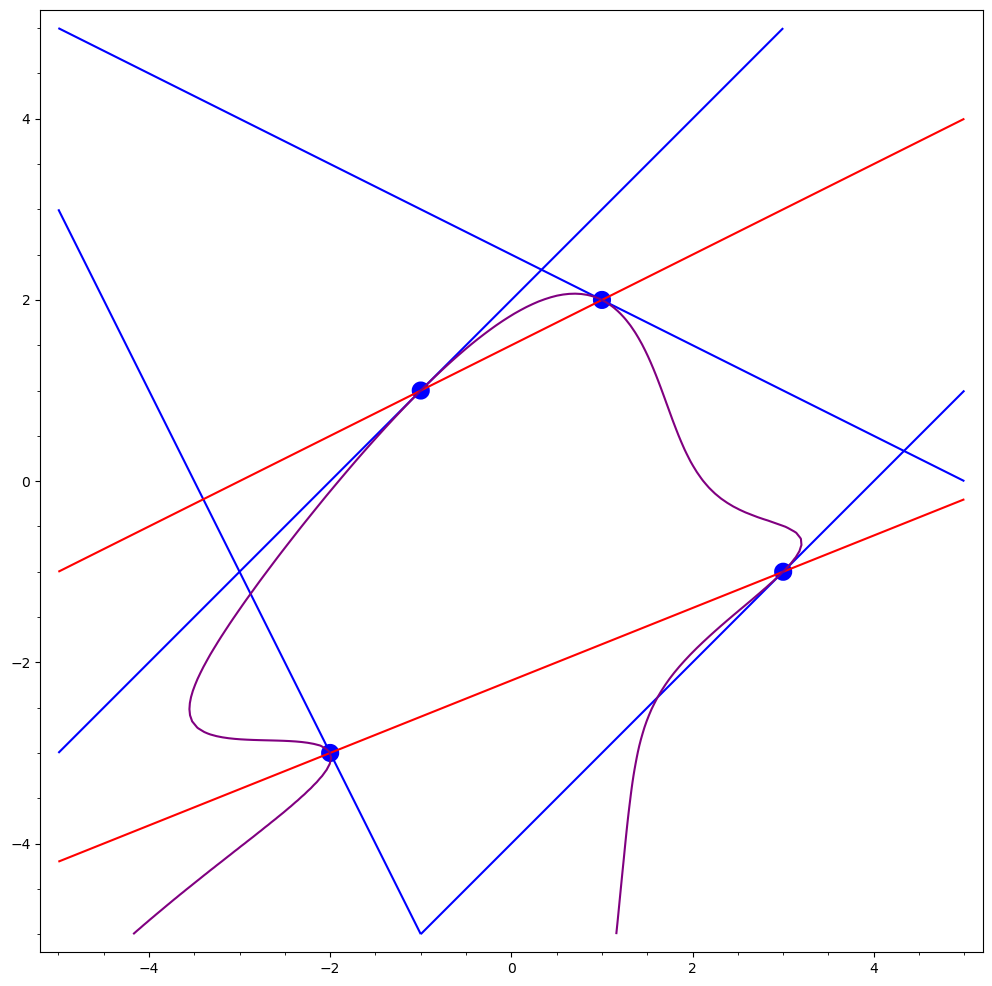}
         \caption{$w_1=1.8,w_2=0.4$}
     \end{subfigure}
     \hfill
     \begin{subfigure}[b]{0.48\linewidth}
         \centering
         \includegraphics[width=\linewidth]{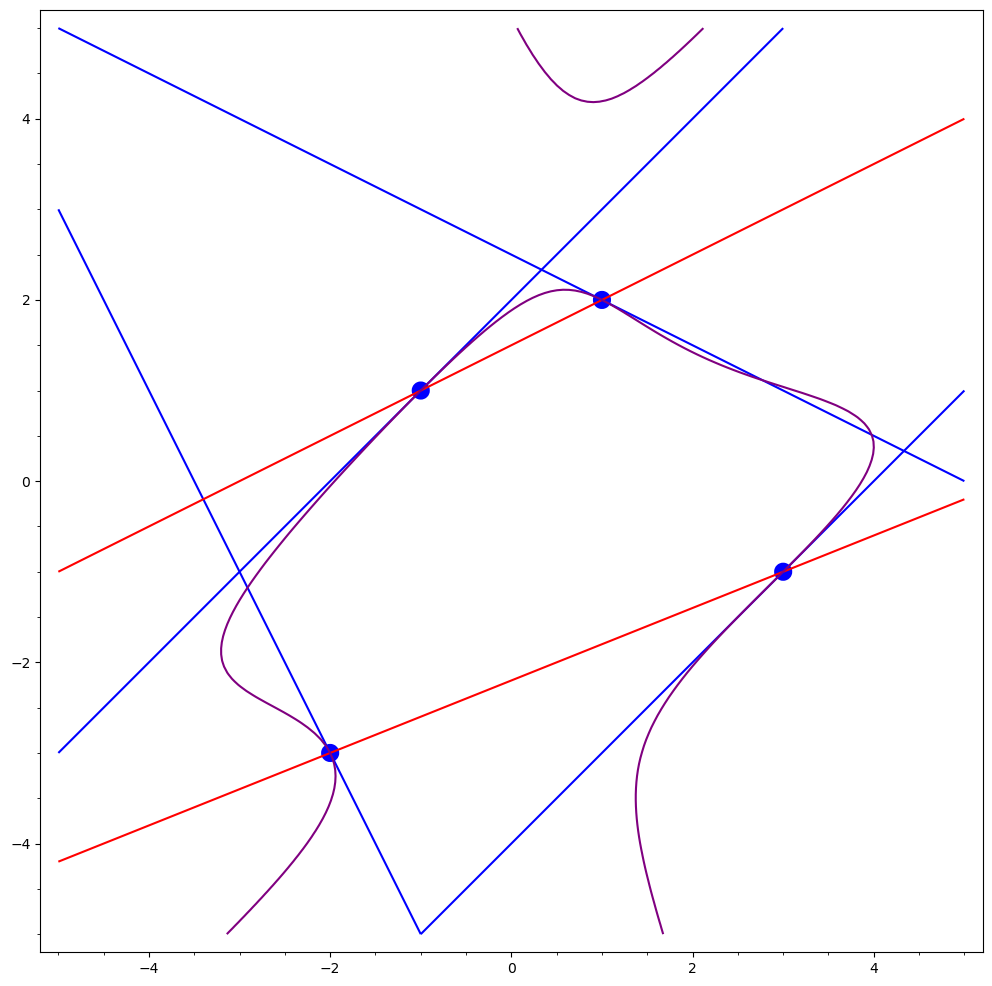}
         \caption{$w_1=2.2,w_2=3.8$}
     \end{subfigure}
        \caption{Different coefficients on a general example}
        \label{fig:other_example}
\end{figure}

In Figure \ref{fig:other_example} you can see an example with general tangent lines and six possible settings of the coefficients.

\section*{Acknowledgments}

This work was supported by the Hungarian Scientific Research Fund (OTKA, No. 124727: Modeling general topology free-form surfaces in 3D). The author thanks Tamás Várady for consultations and valuable comments.

\appendix

\section{Equation system of constraints}
\label{app:eq}

According to Bajaj and Ihm\cite{Bajaj:1992:Hermite} a tangential constraint in 2D can be described with the following equations ($(x,y)$ is the point to interpolate, $(dx,dy)$ is the prescribed gradient).

\begin{align}
    f(x,y) &= 0\\
    dx \cdot \partial_y f(x,y) - dy \cdot \partial_x f(x,y) &= 0
\end{align}

Now, apply this to the unknown quadratic curve $f(x,y) = a \cdot x^2 + b \cdot x y + c \cdot y^2 + d \cdot x + e \cdot y + f$, which should interpolate point $(x_1,y_1)$ with gradient $(m_1, n_1)$, point $(x_2,y_2)$ with gradient $(m_2,n_2)$ and point $(x_3,y_3)$ (with no specified gradient).

The partial derivatives are
\begin{align}
    \partial_x f(x,y) &= 2a \cdot x+b \cdot y+d\\
    \partial_y f(x,y) &= 2c \cdot y+b \cdot x+e
\end{align}

Thus, the equation system will be
\begin{align}
    a x_1^2 + b x_1 y_1 + c y_1^2 + d x_1 + e y_1 + f &= 0\\
    a x_2^2 + b x_2 y_2 + c y_2^2 + d x_2 + e y_2 + f &= 0\\
    a x_3^2 + b x_3 y_3 + c y_3^2 + d x_3 + e y_3 + f &= 0\\
    m_1\cdot(2c y_1+b x_1+e) - n_1\cdot(2a x_1+b y_1+d) &= 0\\
    m_2\cdot(2c y_2+b x_2+e) - n_1\cdot(2a x_2+b y_2+d) &= 0
\end{align}

These are 5 equations for $(a,b,c,d,e,f)$, so the solution space is indeed $6-5=1$ dimensional.

\bibliographystyle{unsrt}
\bibliography{sample}

\end{document}